\documentclass[a4paper,twoside]{amsart}
\usepackage{amsmath}
\usepackage{amsfonts}
\usepackage{amssymb}
\usepackage{amsthm}
\usepackage{newlfont}
\usepackage{graphicx}
\usepackage{amscd}
\usepackage{young}

\theoremstyle{plain}
\newtheorem{Th}{Theorem}[section]
\newtheorem{Cor}[Th]{Corollary}

\newtheorem{Prop}[Th]{Proposition}
\theoremstyle{definition}
\newtheorem{Def}{Definition}[section]
\newtheorem{Ex}{Example}[section]
\theoremstyle{remark}
\newtheorem*{Rem}{Remark}
\numberwithin{equation}{section}

\newcommand{\DD}{{\mathbb D}}

\newcommand{\ZZ}{{\mathbb Z}}

\begin{document}

\title
{Integrability and geometry of the Wynn recurrence}

\author{Adam Doliwa and Artur Siemaszko}

\address{Faculty of Mathematics and Computer Science\\
University of Warmia and Mazury in Olsztyn\\
ul.~S{\l}oneczna~54\\ 10-710~Olsztyn\\ Poland} 
\email{doliwa@matman.uwm.edu.pl,  artur@uwm.edu.pl}

%
\date{}
\keywords{Pad\'{e} approximation, numerical analysis, discrete integrable systems, incidence geometry, quasideterminants, non-commutative rational functions, Fibonacci language, circle packings, discrete analytic functions}
\subjclass[2010]{41A21, 37N30, 37K20, 37K60, 65Q30, 51A20, 42C05}

\begin{abstract}
We show that the Wynn recurrence (the missing identity of Frobenius of the Pad\'{e} approximation theory) can be incorporated into the theory of integrable systems as a reduction of the discrete Schwarzian Kadomtsev--Petviashvili equation. This allows, in particular, to present the geometric meaning of the recurrence as a construction of the appropriately constrained quadrangular set of points.  The interpretation is valid for a projective line over arbitrary skew field what motivates to consider non-commutative Pad\'{e} theory. We transfer the corresponding elements, including the Frobenius identities, to the non-commutative level using the quasideterminants. Using an example of the characteristic series of the Fibonacci language we present an application of the theory to the regular languages. We introduce the non-commutative version of the discrete-time Toda lattice equations together with their integrability structure. Finally, we discuss application of the Wynn recurrence in a different context of the geometric theory of discrete analytic functions.
\end{abstract}
\maketitle

\section{Introduction}

The connection between the theory of integrable systems and numerical algorithms related to the Pad\'{e} approximants~\cite{Brezinski-CFPA,Brezinski,Gragg} was observed many times since the first papers on the subject~\cite{Moser,ChCh,GRP,PGR,PGR-LMP,Hirota-1993}, see also more recent works \cite{NagaiTokihiroSatsuma,BrezinskiHeHuRedivo-ZagliaSun,HeHuSunWeniger,Brezinski-CA-XX,BrezinskiRedivo-Zaglia,NagaoYamada}.According to authors of \cite{CHHL}  \emph{the connection between convergence acceleration algorithms and integrable systems brings a different and fresh look to both domains and could be of benefit to them.} Due to the Hankel (or Toeplitz) structure of determinants used in the Pad\'{e} theory constructions, there exists close connection of the subject with the theory of orthogonal polynomials~\cite{Ismail}. Recently, the relation of partial-skew-orthogonal polynomials to integrable lattices with Pfaffian tau-functions was investigated in~\cite{CHHL-CMP}.

It is well known (see for example recent textbook~\cite{IDS} where the relationship of discrete integrable systems to Pad\'{e} approximants is discussed in detail) that the relation can be described via the discrete-time Toda lattice~\cite{Hirota-2dT} which  is in fact one of the Frobenius identities~\cite{Baker} found long time ago. The subclass of solutions of the discrete-time Toda chain equations relevant in the Pad\'{e} theory is given by restriction to half-infinite chain. In the literature there are known also other special solutions of the equations, like the finite or periodic chain reductions, or infinite chain soliton solutions. Even more developed theory exists for its continuous time version~\cite{Toda-TL,Moser,Berezansky}.

The aforementioned Toda chain equations can be obtained as a particular integrable reduction of the Hirota discrete Kadomtsev--Petviashvili (KP) system~\cite{Hirota}, published initially under the name of \emph{discrete analogue of a generalized Toda equation}. Such a reduction can be applied also on the finite field level~\cite{Bialecki-1DT,BialeckiDoliwa}. We would like to stress that the Hirota system plays distinguished role in the theory of integrable systems and their applications~\cite{Miwa,Shiota,FWN-Capel,KNS-rev}. In particular, it is known that majority of known integrable systems can be obtained as its reductions and/or appropriate continuous limits~\cite{Zabrodin} or subsystems \cite{Doliwa-Des-red,doliwakp}. Moreover, the non-commutative version of the Hirota system~\cite{Nimmo-NCHM}, its $A$-root lattice symmetry structure~\cite{Dol-AN} and geometric interpretation~\cite{Dol-Des} allow to study integrable systems from additional perspectives. The geometric approach to discrete soliton equations initiated in \cite{BP1,DS-AL,BP2,DCN,MQL}, see also \cite{BobSur}, often allows to formulate crucial properties of such systems in terms of incidence geometry statements and provides a meaning for various involved calculations. The integrable systems in non-commutative variables are of special interest recent years~\cite{EGR,Kupershmidt,BobenkoSuris-NC,DMH,RetakhRubtsov,GilsonNimmo,Kondo,LiNimmo,DoliwaKashaev,DoliwaNoumi}. They link, in particular, the classical integrable systems with quantum integrability. 

Having in mind the \emph{substantial practical advantage}~\cite{Baker} of the Wynn recurrence~\cite{Wynn} in the theory of Pad\'{e} approximants, we asked initially about its place within the integrable systems theory. In~\cite{IDS}, in regard for this subject its authors refer to the paper \cite{KR-M-X}, where the recurrence is obtained by application of the reduction group method to the Lax--Darboux schemes associated with nonlinear
Schr\"{o}dinger type equations. We would like also to point  
another recent work~\cite{Kels} where multidimensional consistency approach was applied, among others, to the Wynn recurrence.

Our answer for the question is that the recurrence can be obtained, in full analogy with the derivation of the discrete-time Toda equations from the Hirota system, as a symmetry reduction of the discrete KP equation in its Schwarzian form. Having known the geometric meaning of the discrete Schwarzian KP equation~\cite{DoliwaKosiorek} this result points out towards the corresponding geometric interpretation of the recurrence itself. Because the projective geometric meaning is valid also in the arbitrary skew field case, then one can ask about validity of the recurrence and the Pad\'{e} approximation for the non-commutative series and the corresponding non-commutative rational functions. Fortunately, a substantial part of such theory can be found in the literature~\cite{Draux-rev,Draux-OP-PA,Draux}. Moreover, studies of the non-commutative Pad\'{e} approximants with the help of quasideterminants, which replace determinants in the non-commutative linear algebra, have been  initiated already~\cite{Quasideterminants-GR1,Quasideterminants,NCSF}. We supplement this approach by deriving in such a formalism the non-commutative analogs of the basic Frobenius identities. This allows us to construct the non-commutative version of the discrete-time Toda chain equations together with the corresponding linear problem, what enables to investigate  integrability of the system. We therefore obtain the Wynn recurrence within that full integrability scheme. As an application of the non-commutative Pad\'{e} approximants we studied the characteristic series of the Fibonacci language, which is the paradigmatic example of a regular language. We also found that the same reduction of the discrete Schwarzian KP equation, but in the complex field case~\cite{KoSchief-Men}, was investigated in connection with the theory of circle packings and discrete analytic functions~\cite{Schramm,BobenkoPinkall-DSIS}. Our paper connects therefore two approximation problems, whose relation was not known before. We supplement also previously known generation of the packings by discussion of the generic initial boundary data and related consistency of the construction to the tangential Miquel theorem.

The structure of the paper is as follows. In Section~\ref{sec:q-det} we present the non-commutative Pad\'{e} theory in terms of quasideterminants and we give its application to the characteristic series of the Fibonacci language. In the first part of Section~\ref{sec:q-det-F} we derive, staying still within the quasideterminantal formalism, the non-commutative analogs of basic Frobenius identities. Then we abandon such particular interpretation of the equations and study them within more general context of lattice integrable systems. They become the non-commutative discrete-time Toda chain equations and their linear problem. Section~\ref{sec:nc-W-r} is devoted exclusively to the non-commutative Wynn recurrence. We first show how its solution can be constructed from solutions of the linear problem introduced previously. Then we present its geometric meaning within the context of projective line over arbitrary skew field. We also show how it can be derived as a dimensional dimensional symmetry reduction of the non-commutative discrete Kadomtsev--Petviashvili equation in its Schwarzian form. In the final Section~\ref{sec:W-circle} we present conformal geometry meaning of the complex field version of the Wynn recurrence. 

Throughout the paper we assume that the Reader knows basic elements of the Pad\'{e} approximation theory, as covered for example in the first part of~\cite{Baker}.

\section{Non-commutative Pad\'{e} approximants and quasideterminants} \label{sec:q-det}
Pad\'{e} approximants of series in non-commuting symbols were studied in \cite{Draux-rev,Draux-OP-PA,Draux}. The analogy with the commutative case became even more direct when the quasideterminants~\cite{Quasideterminants-GR1,Quasideterminants}, which replace the standard determinants in the non-commutative linear algebra, have been applied \cite{NCSF} to investigate their properties. We first recall the relevant  properties of the quasideterminants needed in that context. Then we give an application of the theory to study regular languages on example of the characteristic series of the Fibonacci language. 

\subsection{Quasideterminants}
\label{sec:CF-Q}
In this Section we recall, following~\cite{Quasideterminants-GR1}  the definition and basic properties of quasideterminants. 
\begin{Def}
	Given square matrix $X=(x_{ij})_{i,j=1,\dots,n}$ with formal entries $x_{ij}$. In the free division ring~\cite{Cohn} generated by the set $\{ x_{ij}\}_{i,j=1,\dots,n}$ consider the formal inverse matrix $Y=X^{-1}= (y_{ij})_{i,j=1,\dots,n}$ to $X$.
	The $(i,j)$th quasideterminant $|X|_{ij}$ of $X$ is the inverse $(y_{ji})^{-1}$ of the $(j,i)$th element of $Y$.
\end{Def}
Quasideterminants can be computed using the following recurrence relation. For $n\geq 2$ let $X^{ij}$ be the square matrix obtained from $X$ by deleting the $i$th row and the $j$th column (with index $i/j$ skipped from the row/column enumeration), then
\begin{equation} \label{eq:QD-exp}
|X|_{ij} = 
x_{ij} - \sum_{\substack{ i^\prime \neq i \\ j^\prime \neq j }} x_{i j^\prime} (|X^{ij}|_{i^\prime j^\prime })^{-1} x_{i^\prime j}
\end{equation}
provided all terms in the right-hand side are defined.
Sometimes it is convenient to use the following more explicit notation
\begin{equation}
|X|_{ij} = \left| \begin{matrix}
x_{11} & \cdots & x_{1j} & \cdots & x_{1n} \\
\vdots &        & \vdots &        & \vdots \\
x_{i1} & \cdots & \boxed{x_{ij}} & \cdots & x_{in}  \\
\vdots &        & \vdots &        & \vdots \\
x_{n1} & \cdots & x_{nj} & \cdots & x_{nn} 
\end{matrix} \right| .
\end{equation}
\begin{Ex} \label{ex:qd-n=2}
In the simplest case of $n=2$ we have
	\begin{equation} \label{eq:qd-n=2}
	\begin{vmatrix}
	\boxed{x_{11}}& x_{12} \\
	x_{21}& x_{22}	\end{vmatrix} = x_{11} - x_{12} x_{22}^{-1} x_{21} .
	\end{equation} 
\end{Ex}

To study non-commutative Pad\'{e} approximants we will need several properties of quasideterminants, which we list below:
\begin{itemize}
	\item row and column operations,
	\item homological relations,
	\item Sylvester's identity.
\end{itemize}
\subsubsection{Row and column operations}

(i) The quasideterminant $|X|_{ij}$ does not depend on permutations of rows and columns in the matrix $X$ that do not involve the $i$th row and the $j$th column.

(ii) Let the matrix $\tilde{X}$ be obtained from the matrix $X$ by multiplying the $k$th row by the element $\lambda$ of the division ring from the left, then 
	\begin{equation}
	|\tilde{X}|_{ij} = \begin{cases} \lambda|X|_{ij} & \text{if} \quad i = k, \\
	|X|_{ij} & \text{if} \quad i\neq k \quad \text{and} \; \lambda \; \text{is invertible} . \end{cases}
	\end{equation}
	
	(iii) Let the matrix $\hat{X}$ be obtained from the matrix $X$ by multiplying the $k$th column by the element $\mu$ of the division ring from the right, then 
	\begin{equation}
	|\hat{X}|_{ij} = \begin{cases} |X|_{ik} \, \mu & \text{if} \quad j = k, \\
	|X|_{ij} & \text{if} \quad j\neq k \quad \text{and} \; \mu \; \text{is invertible} . \end{cases}
	\end{equation}

(iv) 
Let the matrix $\bar{X}$ is constructed by adding to some row of the matrix $X$ its $k$th row multiplied by a scalar $\lambda$ from the left, then
\begin{equation}
|X|_{ij} = |\bar{X}|_{ij}, \qquad i = 1, \dots , k-1, k+1, \dots , n, \quad j=1,\dots , n.
\end{equation}

 (v) Let the matrix $\check{X}$ is constructed by addition to some column of the matrix $X$ its $l$th column multiplied by a scalar $\mu$ from the right, then
\begin{equation}
|X|_{ij} = |\check{X}|_{ij}, \qquad i = 1, \dots , n, \quad j=1,\dots , l-1, l+1 , \dots ,n.
\end{equation}

\subsubsection{Homological relations}

(i) Row homological relations:
\begin{equation}
-|X|_{ij} \cdot |X^{i k}|_{sj}^{-1} = |X|_{ik} \cdot |X^{ij}|_{sk}^{-1}, \qquad s\neq i.
\end{equation}

(ii) Column homological relations:
\begin{equation} \label{eq:chr}
- |X^{kj}|_{is}^{-1} \cdot |X|_{ij}=  |X^{ij}|_{ks}^{-1} \cdot |X|_{kj} , \qquad s\neq j.
\end{equation}

\subsubsection{Sylvester's identity}

Let $X_0 = (x_{ij})$, $i,j = 1,\dots ,k$, be a submatrix of $X$ that is invertible. For $p,q = k+1,\dots ,n$, set
\begin{equation*}
c_{pq} = \begin{vmatrix}
 &&& x_{1q} \\
& X_0 & & \vdots\\
&&& x_{kq} \\
x_{p1} & \dots & x_{pk} & \boxed{x_{pq}} 
\end{vmatrix} \; ,
\end{equation*}
and consider the $(n-k) \times (n-k)$ matrix $C = (c_{pq})$, $p,q = k+1,\dots , n$. Then for $i,j = k+1,\dots , n$,
\begin{equation}
|X|_{ij} = |C|_{ij} \; .
\end{equation}
In applications Sylvester's identity is usually used in conjunction with row/column permutations.
\begin{Ex}
	Let us take $n=3$ and $k=1$ then using the basic equation~\eqref{eq:qd-n=2} of Example~\ref{ex:qd-n=2} we get
\begin{gather*}
|X|_{33} = \begin{vmatrix}
x_{11}& x_{12} &x_{13} \\
x_{21}& x_{22} & x_{23} \\
x_{31}& x_{32} & \boxed{x_{33}} 
\end{vmatrix} =
\begin{vmatrix}
\begin{vmatrix}
	x_{11} &x_{12} \\
	x_{21}& \boxed{x_{22}} 
	\end{vmatrix} & \begin{vmatrix}
x_{11} &x_{13} \\
x_{21}& \boxed{x_{23} }
\end{vmatrix}  \\ 
\begin{vmatrix}
x_{11} &x_{12} \\
x_{31}& \boxed{x_{32} }
\end{vmatrix} & \boxed{\begin{vmatrix}
x_{11} &x_{13} \\
x_{31}& \boxed{x_{33}}
\end{vmatrix}}
\end{vmatrix}
= \\ 
\begin{vmatrix}
x_{11} &x_{13} \\
x_{31}& \boxed{x_{33}} 
\end{vmatrix} -
\begin{vmatrix}
x_{11} &x_{12} \\
x_{31}& \boxed{x_{32}} 
\end{vmatrix} 
\begin{vmatrix}
x_{11} &x_{12} \\
x_{21}& \boxed{x_{22} }
\end{vmatrix}^{-1} 
\begin{vmatrix}
x_{11} &x_{13} \\
x_{21}& \boxed{x_{23} }
\end{vmatrix}.
\end{gather*} 
\end{Ex}

\subsection{Pad\'{e} approximants of non-commutative series in terms of quasideterminants} \label{sec:ncP}
The results presented below were given in  \cite{NCSF}.
Consider a non-commutative formal series
\begin{equation}
F(t) = S_0 + S_1 t + S_2 t^2 + \dots + S_n t^n + \dots ,
\end{equation}
where the parameter $t$ commutes with the coefficients $S_i$, $i=0,1,2,.\dots $. Set
\begin{equation}
F_k(t) = S_0 + S_1 t + \dots + S_k t^k,
\end{equation}
as its  polynomial of degree $k$ truncation, where also by definition $F_k(t) = 0 $ for $k<0$. Define polynomials 
\begin{equation} \label{eq:PQ-P}
P_{m,n}(t) = \left| \begin{matrix}
\boxed{F_m(t)} & S_{m+1} & \cdots & S_{m+n} \\
t F_{m-1}(t) & S_m     & \cdots & S_{m+n-1}  \\
\vdots       &  \vdots & \ddots & \vdots \\
t^n F_{m-n}(t) & S_{m-n+1} &\cdots & S_m 
\end{matrix} \right| , \quad
Q_{m,n}(t) = \left| \begin{matrix}
\boxed{1}     & S_{m+1} & \cdots & S_{m+n} \\
t & S_m     & \cdots & S_{m+n-1}  \\
\vdots       &  \vdots & \ddots & \vdots \\
t^n  & S_{m-n+1} &\cdots & S_m 
\end{matrix} \right| ,
\end{equation}
of degrees $m$ and $n$ in parameter $t$, respectively. Then the fraction $Q_{m,n}(t)^{-1} P_{m,n}(t) = [m/n]_L$ agrees with $F(t)$ up to terms of order $m+n$ inclusively
\begin{equation} \label{eq:QFP}
Q_{m,n}(t) F(t) = P_{m,n}(t) + O(t^{m+n+1}).
\end{equation}
The proof is based on expanding the left hand side of \eqref{eq:QFP} in powers of $t$ and checking that the coefficients, by formula \eqref{eq:QD-exp}, can be nicely written in terms of certain quasideterminants. The quasideterminants, which multiply $t^k$, $k=0,\dots , m$, coincide with the corresponding coefficients of $P_{m,n}(t)$, while the quasideterminants which multiply $t^k$, $k=m+1,\dots , m+n$,  vanish (their matrices have two columns identical).
\begin{Rem}
There exist analogous quasideterminantal expressions for Pad\'{e} approximants $[m/n]_R$ with denominators on the right side. The matrices of quasideterminants representing the polynomials of the new nominator and denominator are transpose of those given by \eqref{eq:PQ-P}.
\end{Rem}

\subsection{The Fibonacci language} \label{sec:Fib}
In the theory of formal languages one of standard techniques is to investigate series in non-commuting variables~\cite{Salomaa}. In particular, it is known~\cite{Berstel-Reutenauer} that the so called regular languages, i.e. the languages recognized by finite state automata~\cite{Sakarovitch}, give rise to non-commutative rational series. In this Section we consider the Pad\'{e} approximation to the characteristic series of the so called Fibonacci language recovering its rational function representation. The reasoning can be, in principle, transferred to other regular languages.

Consider the language over alphabet $\{ a,b \}$ consisting of words with two consecutive letters $b$ prohibited. Its characteristic series 
\begin{equation} \label{eq:Fib-ser}
F = 1 + a + b + aa + ab + ba + aaa + aab + aba + baa + bab + \dots
\end{equation} 
where $1$ represents the empty word, can be read-out from the corresponding deterministic finite state automaton \cite{Berstel-Reutenauer,Sakarovitch} visualized on Figure~\ref{fig:Fibonacci-a}.
\begin{figure}[h!]
	\includegraphics[width=4cm]{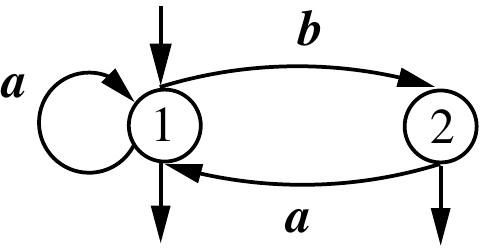} 
	\caption{Deterministic finite state automaton accepting the Fibonacci language}
	\label{fig:Fibonacci-a}
\end{figure}

In terms of the initial state labeled by $1$, the final states $1$ and $2$, and the transition matrix whose elements are given by labels of the edges of the automaton graph, the characteristic series is given by
	\begin{equation}
F = \left( \begin{array}{cc} 1 & 0 \end{array} \right)
\left( \begin{array}{cc} a & b \\ a & 0 \end{array} \right)^* \left( \begin{array}{c} 1  \\ 1 \end{array} \right) ,
\end{equation}
were for a square matrix $M$ its Kleene's star $M^*$ is defined as
\begin{equation}
M^* = I + M + M^2 + M^3 + \dots = (I-M)^{-1}.
\end{equation}
The series can be therefore represented by the corresponding non-commutative rational function expression 
\begin{equation} \label{eq:Fib-rat}
F = [1-(a+ba)]^{-1}(1+b).
\end{equation}

Recall that the number $f_k$ of the Fibonacci words of length $k$ satisfies the well known recurrence relation 
\begin{equation} \label{eq:fib-rec}
f_0 = 1, \quad f_1 = 2, \quad f_{k} = f_{k-1} + f_{k-2}, \qquad k>1.
\end{equation}
The corresponding generating function 
\begin{equation}
f(t) = \sum_{k=0}^\infty f_k t^k,
\end{equation}
which can be found using the recurrence \eqref{eq:fib-rec}, reads
\begin{equation}
f(t) = \frac{1+t}{1-t-t^2} = 1 + 2t + 3 t^2 + 5 t^3 + \dots \, ,
\end{equation} 
and coincides with $[1/2]$ Pad\'{e} approximation of the series. 

\begin{Prop}
	The non-commutative Pad\'{e} approximation technique when applied to the characteristic series of the Fibonacci language \eqref{eq:Fib-ser} allows to recover its non-commutative rational function representation \eqref{eq:Fib-rat}. 
\end{Prop}
\begin{proof}
In order to apply the Pad\'{e} approximation technique in quasideterminantal formalism let us split the characteristic series into homogeneous terms by the transformation $a\mapsto at$, $b\mapsto bt$ what gives the formal series $F(t)$ with coefficients $S_k$ being the formal sum of the Fibonacci words of length $k$. They satisfy the following recurrence relations (which can be also read off from the automaton)
\begin{equation} \label{eq:rec-Fib}
S_0 = 1, \quad S_1 = a+b, \quad S_k = a S_{k-1} + ba S_{k-2} = S_{k-1} a + S_{k-2}ab, \qquad k> 1, 
\end{equation}
i.e., after/before the letter $a$ there can be $a$ or $b$, while after/before the letter $b$ there can be only the letter~$a$.

Using the row and column operations applied to the quasideterminants \eqref{eq:PQ-P} representing the nominator and the denominator of the left Pad\'{e} approximant $[1/2]_L$ we obtain
\begin{equation}
P_{1,2}(t) = 1 + bt, \qquad Q_{1,2}(t) = 1 - at - ba t^2.
\end{equation}
Their ratio agrees with the series $F(t)$ due to the formula (for more details on topology in the space of formal series see~\cite{Sakarovitch})
\begin{equation}
(1 - at - ba t^2)^{-1} (1 + bt) = F_n(t) +
(1 - at - ba t^2)^{-1}(t^{n+1}S_{n+1} + t^{n+2}ba S_n),\quad n\geq 0,
\end{equation}
what can be shown inductively using the recurrence \eqref{eq:rec-Fib}.
\end{proof}

\begin{Cor} Because we started with the rational series then $[m/n]_L = [1/2]_L$ for $m\geq 1$ and $n\geq 2$. In particular, formulas \eqref{eq:PQ-P} give
	\begin{align}
	P_{m,2}(t) = (-1)^{m-1} P_{1,2}(t), & \qquad Q_{m,2}(t) = (-1)^{m-1} Q_{1,2}(t), & &  m \geq 1 ,\\
	P_{1,n}(t) = P_{1,2}(t), &\qquad Q_{1,n}(t) = Q_{1,2}(t) , & &n \geq 2, \\
	P_{m,n} (t) =  0, & \qquad Q_{m,n}(t) = 0 , &  m > 1 \quad & \text{and} \quad n >2. 
	\end{align}
\end{Cor}
\begin{Rem}
	Similar results can be obtained for right Pad\'{e} approximants of the series, where
	\begin{equation}
	[1/2]_R = (1+bt)(1-at - ab t^2)^{-1}.
	\end{equation}
\end{Rem}
\begin{Rem}
Fibonacci language is a paradigmatic example of a regular (called also rational) language~\cite{Berstel-Reutenauer,Sakarovitch}. The reasoning above applies, in principle, also to other such languages. 
\end{Rem}

\section{Quasideterminantal analogs of the Frobenius identities} \label{sec:q-det-F}
In this Section we derive certain recurrences which will play the role of the Frobenius identities~\cite{Baker}. 
We use the notation introduced in Section~\ref{sec:ncP}. 
\subsection{Non-commutative discrete-time Toda chain equations}

The central role in what follows is played by the following quasideterminant
\begin{equation}
\rho_{m,n} = \left| \begin{matrix}
\boxed{S_m} & S_{m+1} & \cdots & S_{m+n} \\
S_{m-1} & S_m     & \cdots & S_{m+n-1}  \\
\vdots       &  \vdots & \ddots & \vdots \\
S_{m-n} & S_{m-n+1} &\cdots & S_m 
\end{matrix} \right| \; ,
\end{equation}
which in the \emph{commutative} case reduces to the ratio
\begin{equation} \label{eq:rho-D}
\rho_{m,n} = \frac{\Delta_{m,n+1}}{\Delta_{m,n}}, \qquad 
\Delta_{m,n} = \left| \begin{matrix}
S_m & S_{m+1} & \cdots & S_{m+n-1} \\
S_{m-1} & S_m     & \cdots & S_{m+n-2}  \\
\vdots       &  \vdots & \ddots & \vdots \\
S_{m-n+1} & S_{m-n+2} &\cdots & S_m 
\end{matrix} \right| \; ,
\end{equation}
of the determinants, which play in turn central role in the theory of the Pad\'{e} approximants.
\begin{Th}
	The quasideterminants $\rho_{m,n}$ satisfy the nonlinear equation
	\begin{equation} \label{eq:nc-D-T}
	\rho_{m+1,n} \left( \rho_{m,n-1}^{-1} - \rho_{m,n}^{-1} \right) \rho_{m-1,n} = \rho_{m,n+1} - \rho_{m,n},
	\end{equation}
which in the commutative case is a consequence of the Frobenius identity~\cite{Baker}
\begin{equation} \label{eq:D-T}
\Delta_{m,n}^2 = \Delta_{m+1,n} \Delta_{m-1 ,n} + \Delta_{m,n+1} \Delta_{m,n-1}.
\end{equation}
\end{Th}

\begin{Rem}
	In the theory of integrable systems equation \eqref{eq:D-T} is called the discrete-time Toda chain equation~\cite{Hirota-2dT} in bilinear form.
\end{Rem}

\begin{proof}
	For the purpose of the proof denote $\rho^1_{m,n} = \rho_{m,n}$, while 
\begin{equation*}
\rho^2_{m,n} = \left| \begin{matrix}
S_m  & \cdots & \boxed{S_{m+n}} \\
\vdots       & \ddots & \vdots \\
S_{m-n} &\cdots & S_m 
\end{matrix} \right| \; ,
\quad
\rho^3_{m,n} = \left| \begin{matrix}
S_m  & \cdots & S_{m+n} \\
\vdots       & \ddots & \vdots \\
\boxed{S_{m-n}} &\cdots & S_m 
\end{matrix} \right| \; , \quad
\rho^4_{m,n} = \left| \begin{matrix}
S_m  & \cdots & S_{m+n} \\
\vdots       & \ddots & \vdots \\
S_{m-n} &\cdots & \boxed{S_m }
\end{matrix} \right| \; .
\end{equation*}
Application of Sylvester's identity to quasideterminants of the matrix 
\begin{equation*}
X =  \begin{pmatrix}
S_m & S_{m+1} & \cdots & S_{m+n+1} \\
S_{m-1} & S_m     & \cdots & S_{m+n}  \\
\vdots       &  \vdots & \ddots & \vdots \\
S_{m-n-1} & S_{m-n} &\cdots & S_m 
\end{pmatrix} \; ,
\end{equation*}
with respect to rows $p_1 = 1$, $p_2 = n+2$, and columns 
$q_1 = 1$, $q_2 = n+2$, gives (among others)
\begin{align} \label{eq:rho-1234}
\rho^1_{m,n+1} = & \rho^1_{m,n} - \rho^2_{m+1 , n} \left( \rho^4_{m,n} \right)^{-1} \rho^3_{m-1,n} ,\\
\rho^3_{m,n+1} = & \rho^3_{m-1,n} - \rho^4_{m,n} \left( \rho^2_{m+1,n} \right)^{-1} \rho^1_{m,n}.
\end{align}
Eliminating $\rho^2_{m+1 , n} \left( \rho^4_{m,n} \right)^{-1} $ from the above equations we obtain
\begin{equation} \label{eq:rho-1-3}
\rho^3_{m,n+1} = \rho^3_{m-1,n} \left( \rho^1_{m,n+1} - \rho^1_{m,n} \right)^{-1} \rho^1_{m,n+1} .
\end{equation}

From the other hand, the column homological relations \eqref{eq:chr} when applied to $X$ and indices $i=n+2$, $j=1$, $k=1$ and $s=2$ give
\begin{equation} \label{eq:chr-1-3}
-\left( \rho^3_{m,n}\right)^{-1} \rho^3_{m,n+1} = \left( \rho^1_{m+1,n}\right)^{-1} \rho^1_{m,n+1}.
\end{equation}
Elimination of $\rho^3_{m,n}$ from equations \eqref{eq:rho-1-3} and \eqref{eq:chr-1-3} leads to the non-commutative discrete-time Toda chain equation~\eqref{eq:nc-D-T}. 

Finally, by substituting in the commutative case expression \eqref{eq:rho-D} into equation~\eqref{eq:nc-D-T} we obtain
\begin{equation*}
 \frac{ \Delta_{m,n+1} \Delta_{m,n-1} - \Delta_{m,n}^2}{\Delta_{m+1,n} \Delta_{m-1 ,n}} =
 \frac{ \Delta_{m,n+2} \Delta_{m,n} - \Delta_{m,n+1}^2}{\Delta_{m+1,n+1} \Delta_{m-1 ,n+1}} .
\end{equation*}
which is immediate consequence of the discrete-time Toda chain equation~\eqref{eq:D-T}.
\end{proof}
\begin{Cor}
Elimination of $\rho^1_{m,n}$ from equations \eqref{eq:rho-1-3} and \eqref{eq:chr-1-3} leads to another version of the non-commutative discrete-time Toda chain equation
\begin{equation}
	\rho^3_{m,n} \left( (\rho^3_{m,n-1})^{-1} - (\rho^3_{m+1,n})^{-1} \right) \rho^3_{m,n} = \rho^3_{m,n+1} - \rho^3_{m-1,n}.
\end{equation}	
After the following change of variables
\begin{equation*}
k=-m, \qquad l = m+n, \qquad \rho^3_{m,n} = H^{(l)}_k,
\end{equation*}	
we get the equation
\begin{equation}
H^{(l)}_k \left[ (H^{(l-1)}_k)^{-1} - (H^{(l+1)}_{k-1})^{-1} \right] H^{(l)}_k = H^{(l+1)}_{k} - H^{(l-1)}_{k+1},
\end{equation}
obtained in \cite{Shi-HaoLi} in the context of the theory of matrix orthogonal polynomials.
\end{Cor}
\begin{Rem}
	Also other quasideterminants $\rho^i_{m,n}$, $i = 2,3,4$, in the commutative case reduce to ratios of the  determinants $\Delta_{m,n}$, in particular
	\begin{equation}
	\rho^3_{m,n} = (-1)^{n+2} \; \frac{\Delta_{m,n+1}}{\Delta_{m+1,n}}.
	\end{equation}
\end{Rem}
\begin{Ex}
	Going back to the Fibonacci language of Section~\ref{sec:Fib} one can check that in such case
	$\rho_{1,1} = b^2(a+b)^{-1}$, $\rho_{1,2+\ell} = \rho_{1,2} = b$ for $\ell\geq 0$, and $\rho_{1+k,1+\ell}= 0$ for $k>0$ and $\ell>0$.
\end{Ex}

\subsection{Non-commutative analogs of Frobenius identities involving the polynomials}
Other Frobenius identities of the classical commutative theory involve the nominator $p_{m,n}(t)$ or denominator $q_{m,n}(t)$ of the Pad\'{e} approximants
\begin{equation}
 p_{m,n}(t) = \left| \begin{matrix}
F_m(t) & S_{m+1} & \cdots & S_{m+n} \\
t F_{m-1}(t) & S_m     & \cdots & S_{m+n-1}  \\
\vdots       &  \vdots & \ddots & \vdots \\
t^n F_{m-n}(t) & S_{m-n+1} &\cdots & S_m 
\end{matrix} \right| , 
  \quad q_{m,n}(t) = \left| \begin{matrix}
1    & S_{m+1} & \cdots & S_{m+n} \\
t & S_m     & \cdots & S_{m+n-1}  \\
\vdots       &  \vdots & \ddots & \vdots \\
t^n  & S_{m-n+1} &\cdots & S_m 
\end{matrix} \right| .
\end{equation}
Their relation to the previous expressions follows from definition of the quasideterminants (for commuting symbols) and are given by
\begin{equation}
p_{m,n}(t) = P_{m,n}(t) \Delta_{m,n}, \qquad
q_{m,n}(t) = Q_{m,n}(t) \Delta_{m,n}.
\end{equation}
Below we present the non-commutative variants of the Frobenius identities 
\begin{align}
\Delta_{m,n+1} w_{m+1,n+1}(t) -  \Delta_{m+1,n+1} w_{m,n+1}(t) & = t \Delta_{m+1,n+2} w_{m,n}(t), \\
\Delta_{m,n+1} w_{m+1,n}(t) -  \Delta_{m+1,n}w_{m,n+1}(t) & = t \Delta_{m+1,n+1} w_{m,n}(t),
\end{align}
satisfied by any linear combination $w_{m,n}(t)$ of $p_{m,n}(t)$ and $q_{m,n}(t)$.
\begin{Th}
	For arbitrary non-commutative formal power series $\lambda(t)$, and $\mu(t)$ the (right) linear combination 
	\begin{equation}
	W_{m,n}(t) = P_{m,n}(t) \lambda(t) + Q_{m,n}(t) \mu(t),
	\end{equation}
satisfies equations
\begin{align} \label{eq:F-W-1t}
W_{m+1,n+1}(t) - W_{m,n+1}(t)  & =  t \rho_{m+1,n+1}\rho^{-1}_{m,n} W_{m,n}(t),\\
\label{eq:F-W-2t}
W_{m+1,n}(t) - W_{m,n+1}(t) & =  t \rho_{m+1,n}\rho^{-1}_{m,n} W_{m,n}(t) .
\end{align}
\end{Th}
\begin{proof}
We will show the result for $W_{m,n}(t)$ being either $P_{m,n}(t)$ or $Q_{m,n}(t)$, because then the conclusion follows from linearity of the equations. We can write
\begin{equation} \label{eq:W-P}
W_{m,n}(t) = \left| \begin{matrix}
\boxed{C_m(t)} & S_{m+1} & \cdots & S_{m+n} \\
t C_{m-1}(t) & S_m     & \cdots & S_{m+n-1}  \\
\vdots       &  \vdots & \ddots & \vdots \\
t^n C_{m-n}(t) & S_{m-n+1} &\cdots & S_m 
\end{matrix} \right|,
\end{equation}
where  $C_k(t) = F_k(t)$ when $W_{m,n}(t) = P_{m,n}(t)$, and $C_k(t) = 1$ when $W_{m,n}(t) = Q_{m,n}(t)$, for all $k\in\ZZ$. The possibility of applying the same proof follows form the observation that in both cases we have 
\begin{equation} \label{eq:W-k-P}
W_{m,n}(t) = \left| \begin{matrix}
\boxed{C_{m+i}(t)} & S_{m+1} & \cdots & S_{m+n} \\
t C_{m-1+i}(t) & S_m     & \cdots & S_{m+n-1}  \\
\vdots       &  \vdots & \ddots & \vdots \\
t^n C_{m-n+i}(t) & S_{m-n+1} &\cdots & S_m 
\end{matrix} \right|, \qquad i = 0, \dots , n.
\end{equation}
For the denominators this is trivial, and for the nominators it can be shown by using column operations.

Application of Sylvester's identity to the matrix 
\begin{equation}
 \begin{pmatrix}
C_{m+1}(t) & S_{m+2} & \cdots & S_{m+n+2} \\
t C_{m}(t) & S_{m+1}    & \cdots & S_{m+n+1}  \\
\vdots       &  \vdots & \ddots & \vdots \\
t^{n+2} C_{m-n}(t) & S_{m-n+1} &\cdots & S_{m+1} 
\end{pmatrix} ,
\end{equation}
with rows $p_1 = 1$, $p_2 = n+2$ and columns $q_1 = 1$, $q_2 = n+2$ gives
\begin{equation}
W_{m+1,n+1}(t) =  W_{m+1,n}(t)  + t \rho^2_{m+2,n}(\rho^{4}_{m+1,n})^{-1} \rho^3_{m,n-1} (\rho^1_{m+1,n-1})^{-1} W_{m,n}(t),
\end{equation}
where the homological column relations were also used. Equation~\eqref{eq:rho-1234} allows then to write
\begin{equation} \label{eq:W-F-1-long}
W_{m+1,n+1}(t) =  W_{m+1,n}(t)  - t (\rho^1_{m+1,n+1} - \rho^{1}_{m+1,n}) (\rho^3_{m,n})^{-1} \rho^3_{m,n-1} (\rho^1_{m+1,n-1})^{-1} W_{m,n}(t).
\end{equation} 

Applying in turn Sylvester's identity to the matrix 
\begin{equation}
\begin{pmatrix}
C_{m+1}(t) & S_{m+1} & \cdots & S_{m+n+1} \\
t C_{m}(t) & S_{m}    & \cdots & S_{m+n}  \\
\vdots       &  \vdots & \ddots & \vdots \\
t^{n+1} C_{m-n}(t) & S_{m-n} &\cdots & S_{m} 
\end{pmatrix} ,
\end{equation}
with rows $p_1 = 1$, $p_2 = n+2$ and columns $q_1 = 1$, $q_2 = 2$ we obtain
\begin{equation} \label{eq:W-F-2}
W_{m,n+1}(t) =  W_{m+1,n}(t)  + t \rho^1_{m+1,n}(\rho^{3}_{m,n})^{-1} \rho^3_{m,n-1} (\rho^1_{m+1,n-1})^{-1} W_{m,n}(t),
\end{equation}
where apart from the homological column relations we used also identity~\eqref{eq:W-k-P} for $i=1$. Subtracting equation \eqref{eq:W-F-2} from \eqref{eq:W-F-1-long} we obtain
\begin{equation} \label{eq:W-F-3}
W_{m+1,n+1}(t) =  W_{m,n+1}(t)  - t \rho^1_{m+1,n+1} (\rho^3_{m,n})^{-1} \rho^3_{m,n-1} (\rho^1_{m+1,n-1})^{-1} W_{m,n}(t).
\end{equation} 
Finally, with the help of the homological relations \eqref{eq:chr-1-3} equations \eqref{eq:W-F-3} and \eqref{eq:W-F-2} can be brought to the form of \eqref{eq:F-W-1t} and \eqref{eq:F-W-2t}, respectively.
\end{proof}
\begin{Cor}
	Equation \eqref{eq:W-F-1-long}, which can be brought to the form
\begin{equation} \label{eq:W-F-1}
W_{m+1,n+1}(t) =  W_{m+1,n}(t)  + t (\rho_{m+1,n+1} - \rho_{m+1,n}) \rho_{m,n}^{-1} W_{m,n}(t),
\end{equation} 
is the non-commutative variant of the Frobenius identity 
\begin{equation}
\Delta_{m+1,n} w_{m+1,n+1}(t) -  \Delta_{m+1,n+1} w_{m+1,n}(t)  = - t \Delta_{m+2,n+1} w_{m,n}(t).
\end{equation}	
\end{Cor}

\subsection{Integrability of the non-commutative discrete time Toda chain equations}
Our approach will be typical to analogous works in the theory of integrable systems. After having derived several identities satisfied by the quasideterminants used to find Pad\'{e} approximants, we abandon such a specific interpretation and consider the equations within more general context of non-commutative integrable systems. 

Motivated by interpretation of the Frobenius identities in the commutative case as the discrete-time Toda chain equations~\cite{Hirota-1993} and the corresponding spectral problem, let us devote this Section to presentation of the non-commutative version of the equation in the formalism known from applications to $\varepsilon$-algorithm~\cite{NagaiTokihiroSatsuma} or orthogonal polynomials~\cite{PGR-LMP}. In the context of non-commutative continued fractions and the corresponding LR-algorithm a non-commutative system of such form was obtained by Wynn~\cite{Wynn-NCF}, while its non-autonomous  generalization for double-sided non-commutative continued fractions was given in \cite{Doliwa-NCCF}.

\begin{Prop} \label{prop:WDE}
	The compatibility of the linear system 
	\begin{align} \label{eq:lin-W1}
	W_{m+1,n+1}(t) & = W_{m,n+1}(t) + tD_{m,n}W_{m,n}(t), \\
	\label{eq:lin-W2}
	W_{m+1,n}(t) & = W_{m,n+1}(t) + t E_{m,n}W_{m,n}(t),
	\end{align} is provided by equations
	\begin{gather} \label{eq:ncT-1}
	D_{m+1,n} E_{m,n} = E_{m+1,n+1} D_{m,n} \\
	\label{eq:ncT-2}
	D_{m+1,n-1} + E_{m,n} = D_{m,n} + E_{m+1,n}.
	\end{gather}
\end{Prop}
\begin{proof}
	Given $W_{m,n}(t)$ and $W_{m+1,n}(t)$ one can find $W_{m+2,n}(t)$ in two different ways:
	\begin{enumerate}
		\item via intermediate steps through $W_{m,n-1}(t)$ using \eqref{eq:lin-W1}
		\begin{equation*}
		W_{m,n-1}(t) = \frac{1}{t} D^{-1}_{m,n-1} ( W_{m+1,n}(t) - W_{m,n}(t)),
		\end{equation*}
		 and through $W_{m+1,n-1}(t)$ using \eqref{eq:lin-W2}
		\begin{equation*}
		W_{m+1,n-1}(t) = W_{m,n}(t) + t E_{m,n-1} W_{m,n-1}(t) = 
		W_{m,n}(t) + E_{m,n-1}D^{-1}_{m,n-1} ( W_{m+1,n}(t) - W_{m,n}(t)),
		\end{equation*}
		what results, due to \eqref{eq:lin-W1}, in
		\begin{equation} \label{eq:W2-1}
		W_{m+2,n}(t) = 
		W_{m+1,n}(t) + t D_{m+1,n-1} \left[ W_{m,n}(t) + E_{m,n-1}D^{-1}_{m,n-1} ( W_{m+1,n}(t) - W_{m,n}(t))\right];
		\end{equation}
		\item via intermediate steps through $W_{m,n+1}(t)$ 
		using \eqref{eq:lin-W2}
		\begin{equation*}
		W_{m,n+1}(t) = W_{m+1,n}(t) - t E_{m,n} W_{m,n}(t),
		\end{equation*}		
		and then through $W_{m+1,n+1}$ using \eqref{eq:lin-W1}
		\begin{equation*}
		W_{m+1,n+1} = 
		W_{m+1,n}(t) + t \left( D_{m,n} - E_{m,n}\right) W_{m,n}(t),
		\end{equation*}
		what results, due to \eqref{eq:lin-W2}, in
		\begin{equation} \label{eq:W2-2}
		W_{m+2,n}(t) = 
		W_{m+1,n}(t) + t \left( D_{m,n} - E_{m,n}\right) W_{m,n}(t) + t E_{m+1,n}  W_{m+1,n}(t).
		\end{equation}
	\end{enumerate}
	Comparison of the coefficients in front of $W_{m,n}(t)$ and $W_{m+1,n}(t)$ in equations \eqref{eq:W2-1} and \eqref{eq:W2-2} gives the statement.
\end{proof}
\begin{Cor} \label{cor:DErho}
	The first equation \eqref{eq:ncT-1} of the nonlinear system can be resolved introducing the potential $\rho_{m,n}$ such that
	\begin{equation}
	D_{m,n} = \rho_{m+1,n+1} \rho^{-1}_{m,n}, \qquad 
	E_{m,n} = \rho_{m+1,n} \rho^{-1}_{m,n},
	\end{equation}
	while the second equation \eqref{eq:ncT-2} gives then \eqref{eq:nc-D-T}.
\end{Cor}
\begin{Rem}
	The above substitution has a meaning in the general context of non-commutative discrete integrable systems, i.e. the potential does not have to come with the quasideterminantal interpretation in the non-commutative Pad\'{e} theory.
\end{Rem}

\section{Non-commutative Wynn recurrence}
\label{sec:nc-W-r}

In \cite{Draux-rev} it was also shown that Pad\'{e} approximants in non-commuting symbols satisfy the Wynn recurrence. In this Section we show that the Wynn recurrence follows from properties of the non-commutative Frobenius identities. Because our result is valid in the more general context of discrete non-commutative integrable systems we will not use the Pad\'{e} table notation. 

We also provide geometric meaning of the recurrence as a relation between five points of a projective line. In the classical Pad\'{e} approximation it will be the projective line over the field of (semiinfinite) Laurent series over the complex or real numbers whose proper subfield is the field of rational functions. 
The geometric meaning of the Wynn recurrence retains its validity also in the non-commutative case, in particular for Mal’cev--Neumann series~\cite{Malcev,Neumann} and the universal ring of fractions by Cohn~\cite{Cohn}.

\subsection{Derivation of the Wynn recurrence} \label{sec:Wynn-derivation}
Notice that the immediate consequence of the linear system \eqref{eq:lin-W1}-\eqref{eq:lin-W2} is the equation
\begin{equation} \label{eq:lin-W3}
W_{m+1,n+1}(t)  = W_{m+1,n}(t) + t C_{m,n}W_{m,n}(t), \qquad \text{where} \quad 
C_{m,n} =  D_{m,n} - E_{m,n}.
\end{equation}

\begin{Prop}
If $ P_{m,n}(t)$ and $Q_{m,n}(t)$ are nontrivial solutions of the linear system \eqref{eq:lin-W1}-\eqref{eq:lin-W2}, then the function $R_{m,n}(t) =  Q^{-1}_{m,n}(t) P_{m,n}(t)$ satisfies the non-commutative Wynn recurrence
\begin{equation} \label{eq:nc-W} \begin{split}
(R_{m+1,n}(t) - R_{m,n}(t) )^{-1} + (R_{m-1,n}(t) - R_{m,n}(t) )^{-1}= \\
(R_{m,n+1}(t) - R_{m,n}(t) )^{-1} + (R_{m,n-1}(t) - R_{m,n}(t) )^{-1}. \end{split}
\end{equation}	
\end{Prop}
\begin{proof}
Inserting such $R_{m,n}(t)$ in the linear system \eqref{eq:lin-W1}-\eqref{eq:lin-W2} and \eqref{eq:lin-W3}
we obtain the following equations
\begin{align*}
Q_{m,n+1}(t) (R_{m,n+1}(t) - R_{m-1,n}(t)) & = 
Q_{m,n}(t) (R_{m,n}(t) - R_{m-1,n}(t)), \\
Q_{m+1,n}(t) (R_{m+1,n}(t) - R_{m,n-1}(t)) & = 
Q_{m,n}(t) (R_{m,n}(t) - R_{m,n-1}(t)), \\
Q_{m+1,n}(t) (R_{m+1,n}(t) - R_{m,n}(t)) & = 
Q_{m,n+1}(t) (R_{m,n+1}(t) - R_{m,n}(t)).
\end{align*}	
Elimination of $Q_{m,n}(t)$ and its shifts leads directly to the equation
\begin{gather}
(R_{m,n}(t) - R_{m-1,n}(t))^{-1} 
(R_{m,n}(t) - R_{m,n-1}(t)) = \\ \nonumber
(R_{m,n+1}(t) - R_{m-1,n}(t))^{-1}
(R_{m,n+1}(t) - R_{m,n}(t))
(R_{m+1,n}(t) - R_{m,n}(t))^{-1}
(R_{m+1,n}(t) - R_{m,n-1}(t)),
\end{gather}
equivalent to the Wynn recurrence.
\end{proof}
\begin{Rem}
	In the standard terminology used in the Pad\'{e} theory, where $R_{m,n}(t)$ is denoted by $C$ (for center) and the other terms $R_{m,n-1}=S$ (south), $R_{m+1,n} = E$ (east), $R_{m-1,n} = W$ (west) and $R_{m,n+1}=N$ (north), the non-commutative Wynn recurrence takes the form \cite{Wynn,Draux-rev}
\begin{equation*}
(E-C)^{-1} + (W-C)^{-1} = (N-C)^{-1} + (S-C)^{-1}.
\end{equation*}
\end{Rem}
\begin{Cor}
Two systems analogous to equations \eqref{eq:ncT-1}-\eqref{eq:ncT-2} which involve other pairs of unknown functions have the form
	\begin{gather} \label{eq:ncT-1-CD}
D_{m+1,n} C_{m,n} = C_{m,n+1} D_{m,n} \\
\label{eq:ncT-2-CD}
D_{m+1,n-1} + C_{m+1,n} = D_{m+1,n} + C_{m,n},
\end{gather}
and
	\begin{gather} \label{eq:ncT-1-CE}
D_{m,n+1} E_{m,n} = E_{m+1,n+1} C_{m,n} \\
\label{eq:ncT-2-CE}
C_{m+1,n} + E_{m+1,n} = C_{m,n+1} + E_{m+1,n+1}.
\end{gather}
\end{Cor}
\begin{Rem}
Notice that the Wynn recurrence is valid in the more general context of the linear problem of the non-commutative discrete-time Toda equations.  	
	In the standard application to Pad\'{e} approximation we are looking for $R_{m,n}(t)$, $m,n\geq 0$ with the initial boundary data of consisting of $R_{m,-1}(t) = \infty$, $R_{m,0}(t) = F_m(t)$, and $R_{-1,n} = 0$. For the general case with $(m,n)\in\mathbb{Z}^2$ as the initial boundary data one can take, for example, the values of $R_{m,-1}$ and $R_{m,0}$ for $m\in\mathbb{Z}$.
	\end{Rem}

All the above results have their "transposed" versions with reversed order of multiplication, what is motivated by the theory of non-commutative right Pad\'{e} approximants (see  Remark at the end of Section~\ref{sec:ncP}). By $\widetilde{W}_{m,n}$, $\widetilde{D}_{m,n}$, $\widetilde{E}_{m,n}$ and $\widetilde{\rho}_{m,n}$ let us denote the corresponding analogs of the functions appearing in Proposition~\ref{prop:WDE} and Corollary~\ref{cor:DErho}, whose "transposed" versions we discuss below in the general context of non-commutative discrete integrable systems (in particular, not restricting ourselves to the non-commutative right Pad\'{e} approximants).
The proof of the following result is left to the Reader as an exercise.
\begin{Prop}
	(i) The compatibility of the linear system 
\begin{align} \label{eq:lin-W1-r}
\widetilde{W}_{m+1,n+1}(t) & = \widetilde{W}_{m,n+1}(t) + t\widetilde{W}_{m,n}(t) \widetilde{D}_{m,n}, \\
\label{eq:lin-W2-r}
\widetilde{W}_{m+1,n}(t) & = \widetilde{W}_{m,n+1}(t) + t \widetilde{W}_{m,n}(t) \widetilde{E}_{m,n},
\end{align} is provided by equations
\begin{gather} \label{eq:ncT-1-r}
\widetilde{E}_{m,n} \widetilde{D}_{m+1,n}  = \widetilde{D}_{m,n} \widetilde{E}_{m+1,n+1}  \\
\label{eq:ncT-2-r}
\widetilde{D}_{m+1,n-1} + \widetilde{E}_{m,n} = \widetilde{D}_{m,n} + \widetilde{E}_{m+1,n}.
\end{gather}
(ii) 	The first equation \eqref{eq:ncT-1-r} of the nonlinear system can be resolved introducing the potential $\widetilde{\rho}_{m,n}$ such that
\begin{equation}
\widetilde{D}_{m,n} = \widetilde{\rho}^{-1}_{m,n} \widetilde{\rho}_{m+1,n+1} , \qquad 
\widetilde{E}_{m,n} =\widetilde{\rho}^{-1}_{m,n} \widetilde{\rho}_{m+1,n} ,
\end{equation}
while the second equation \eqref{eq:ncT-2-r} gives 	\begin{equation} \label{eq:nc-D-T-r}
\widetilde{\rho}_{m-1,n} \left( \widetilde{\rho}_{m,n-1}^{-1} - \widetilde{\rho}_{m,n}^{-1} \right) \widetilde{\rho}_{m+1,n} = \widetilde{\rho}_{m,n+1} - \widetilde{\rho}_{m,n}.
\end{equation}
(iii) If $ \widetilde{P}_{m,n}(t)$ and $\widetilde{Q}_{m,n}(t)$ are nontrivial solutions of the linear system \eqref{eq:lin-W1-r}-\eqref{eq:lin-W2-r}, then the function $\widetilde{R}_{m,n}(t) =  \widetilde{P}_{m,n}(t) \widetilde{Q}^{-1}_{m,n}(t) $ satisfies the non-commutative Wynn recurrence
\eqref{eq:nc-W}.

\end{Prop}

\subsection{Geometry of the Wynn recurrence}
\label{sec:proj-geome}
This Section is devoted to the geometric meaning of the Wynn recurrence interpreted as the relation between five points of a projective line. We consider geometry over arbitrary (including skew) field $\mathbb{D}$. The original case studied by Wynn~\cite{Wynn} dealt with the field of rational functions, however the non-commuting variables were also within his interest~\cite{Wynn-NCF}. Our approach will follow that used recently in \cite{DoliwaKosiorek} to provide geometric meaning of the non-commutative discrete Schwarzian Kadomtsev--Petviashvili equation.

\begin{Prop} \label{prop:Wynn-geom}
	Interpreted as a relation between five points of the projective (base) line the Wynn recurrence~\eqref{eq:nc-W} is equivalent the following construction of the point $R_{m,n+1}$ once the points $R_{m-1,n}$, $R_{m,n-1}$, $R_{m,n}$ and $R_{m+1,n}$ are given (see Figure \ref{fig:Wynn-geometry-12}):
	\begin{figure}[h!]
		\begin{center}
			\includegraphics[width=9cm]{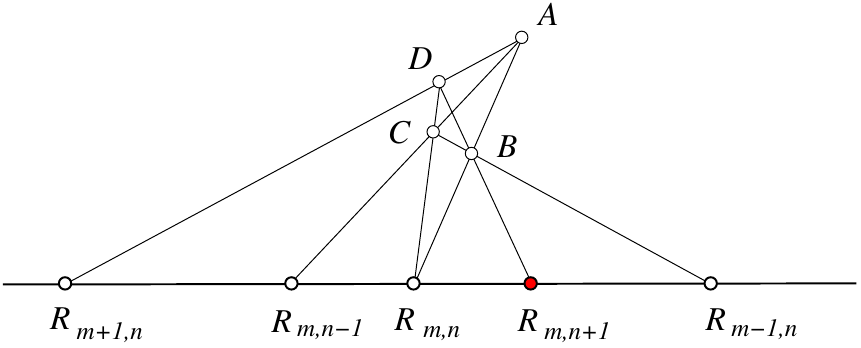}
		\end{center}
		\caption{Geometric meaning of the Wynn recurrence}
		\label{fig:Wynn-geometry-12}
	\end{figure}
	\begin{itemize}
		\item select any point $A$ outside the base line,
		\item on the line $\langle A, R_{m,n} \rangle$ select any point $B$ different from $A$ and $R_{m,n}$,
		\item define point $C$ as the intersection of lines $\langle A, R_{m,n-1} \rangle$ and $\langle B, R_{m-1,n} \rangle$,
		\item define point $D$ as the intersection of lines $\langle C, R_{m,n} \rangle$ and $\langle A, R_{m+1,n} \rangle$,
		\item point $R_{m,n+1}$ is the intersection of the line  $\langle B, D \rangle$ with the base line.
	\end{itemize}
\end{Prop}
\begin{proof}
	The above arbitrariness of the points $A$ and $B$ in the construction is known in the projective geometry~\cite{Coxeter-PG} and follows from  the Desargues theorem. We will use the freedom to simplify the calculation.
	
	Consider the non-homogeneous coordinates where the base line (except from the infinity point) is given by the first coordinate $\{ (x,0) \colon x\in \DD \}$. As  $A$ we choose the point of the infinity line where the lines parallel to the second coordinate line meet.  Then as the point $B$ on the line through $(R_{m,n},0)$ and parallel to that line we take the point $B=(R_{m,n},1)$. From now on there is no freedom in the construction. 
	
	The coordinates of the point $C$ 
	\begin{equation}
	C=(R_{m,n-1}, t), \qquad t =  
	1 - (R_{m,n-1} - R_{m,n})(R_{m-1,n} - R_{m,n})^{-1}	
	\end{equation} 
	can be found from the equation
	\begin{equation*}
	(R_{m,n-1},t) = (R_{m-1,n} , 0 ) + s \left[ (R_{m,n},1) - (R_{m-1,n},0) \right],
	\end{equation*}
	where by the standard convention when representing vectors as rows we multiply them by scalars from the left. Similar calculation gives coordinates of the point $D$
	\begin{equation}
	D=(R_{m+1,n}, t^\prime), \qquad t^\prime =  
	(R_{m+1,n} - R_{m,n})[(R_{m,n-1} - R_{m,n})^{-1} - (R_{m-1,n} - R_{m,n})^{-1}] .	
	\end{equation} 
	Finally, $R_{m,n+1}$ can be calculated from the equation
	\begin{equation*}
	(R_{m,n+1},0) = (R_{m,n} , 1 ) + s^\prime \left[ (R_{m+1,n}, t^\prime) - (R_{m,n},1) \right],
	\end{equation*}
	which gives $s^\prime = (1-t^\prime)^{-1}$ and leads to the Wynn recurrence \eqref{eq:nc-W}.
\end{proof}

\begin{Rem}
	In defining coordinates on projective line~\cite{Coxeter-PG} the above construction provides the additive structure in the (skew) field.\footnote{We thank Jaros{\l}aw Kosiorek for pointing us such a geometric interpretation of the construction.} Indeed, when $R_{m,n}$ is moved to the infinity point, and points $R_{m,n-1}, R_{m-1,n}$, $R_{m+1,n}$ are identified with $0, a, b$, respectively, then $R_{m,n+1}$ represents $a+b$, see Figure~\ref{fig:Wynn-addition} (here parallel lines intersect in the corresponding points $A$, $B$ or $R_{m,n}$ of the infinity line). 
	\begin{figure}[h!]
		\begin{center}
			\includegraphics[width=9cm]{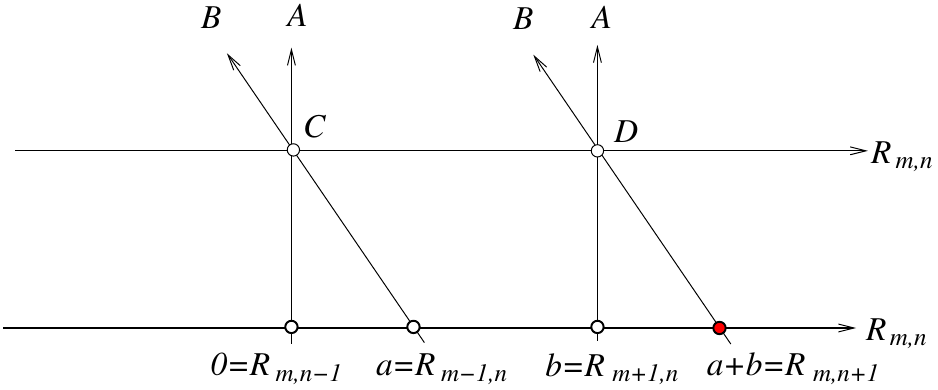}
		\end{center}
		\caption{Additive structure of a projective line}
		\label{fig:Wynn-addition}
	\end{figure}
	Algebraic verification follows from the fact that in the limit $R_{m,n}\to\infty$ the Wynn recurrence reduces to
	\begin{equation}
	R_{m-1,n} + R_{m+1,n} = R_{m,n-1} + R_{m,n+1}, \qquad R_{m,n}=\infty .
	\end{equation}
\end{Rem}

\subsection{Reduction of the non-commutative discrete Schwarzian KP equation} It is well known~\cite{KricheverLipanWiegmannZabrodin,Bialecki-1DT} that the discrete-time Toda chain system~\eqref{eq:D-T} can be obtained as a reduction of the discrete KP equation in its bilinear form~\cite{Hirota-2dT}. Let us derive the non-commutative Wynn recurrence as a corresponding reduction of the non-commutative discrete KP equation in its Schwarzian form~\cite{FWN-Capel,BoKo-N-KP}
\begin{equation} \label{eq:dSKP} \begin{split}
(R_{m,n+1,p+1} - R_{m,n,p+1})(R_{m,n+1,p+1} - R_{m,n+1,p})^{-1}  (R_{m+1,n+1,p} - R_{m,n+1,p}) = \\(R_{m+1,n,p+1} - R_{m,n,p+1}) (R_{m+1,n,p+1} - R_{m+1,n,p})^{-1}
(R_{m+1,n+1,p} - R_{m+1,n,p}),\end{split}
\end{equation}
where $R\colon \ZZ^3 \to \DD$ is an unknown function of three discrete variables. 

\begin{Prop} \label{prop:Wynn-red}
	Assume that the non-commutative discrete Schwarzian Kadomtsev--Petviashvili equation~\eqref{eq:dSKP} is subject to the constraint
	\begin{equation} \label{eq:red-SKP-W}
	R_{m+1,n+1,p} = R_{m,n,p+1},
	\end{equation}
	then the equation reduces to the non-commutative Wynn recurrence~\eqref{eq:nc-W}.
\end{Prop}
\begin{proof}
	Because of the reduction condition the function $R_{m,n,p}$ becomes effectively the function $R_{m,n}$ of two discrete variables.
	Replacing the shift in the third variable in equation~\eqref{eq:dSKP} by simultaneous shifts on the first and second variables we get
	\begin{equation*} \begin{split}
	(R_{m+1,n+1} - R_{m+1,n})^{-1} [
	(R_{m+2,n+1} - R_{m+1,n+1} )- (R_{m+1,n} - R_{m+1,n+1})] 
	(R_{m+2,n+1} - R_{m+1,n+1})^{-1} = \\ 
(R_{m+1,n+1} - R_{m,n+1})^{-1} [
(R_{m+1,n+2} - R_{m+1,n+1} )- (R_{m,n+1} - R_{m+1,n+1})] 
(R_{m+1,n+2} - R_{m+1,n+1})^{-1}  ,
	\end{split}
	\end{equation*}
	which after natural cancellations gives shifted recurrence \eqref{eq:nc-W}.
\end{proof}

\begin{Rem}
	As it was shown in~\cite{DoliwaKosiorek} the non-commutative discrete Schwarzian KP equation can be interpreted as a relation between six points (called the quadrangular set) of the projective line visualized in Figure~\ref{fig:R-geometry-123}. The geometric meaning of the Wynn recurrence described in Proposition~\ref{prop:Wynn-geom} follows then by the application of the reduction condition~\eqref{eq:red-SKP-W}.
\end{Rem}

\begin{figure}[h!]
	\begin{center}
		\includegraphics[width=9cm]{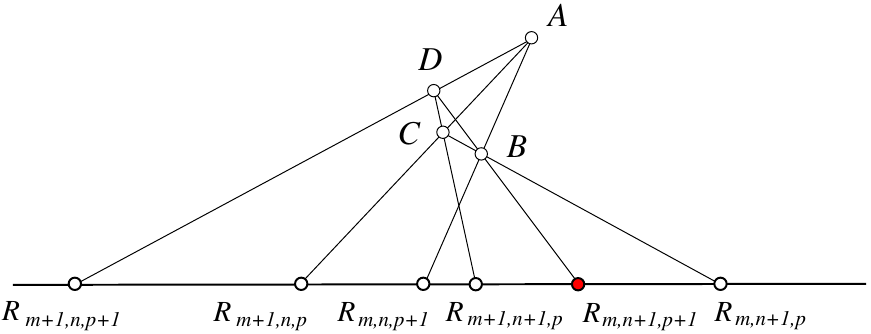}
	\end{center}
	\caption{Geometric meaning of the non-commutative discrete Schwarzian Kadomtsev--Petviashvili equation}
	\label{fig:R-geometry-123}
\end{figure}
\begin{Rem}
	The reduction condition~\eqref{eq:red-SKP-W} expresses invariance to the corresponding solution of the discrete Schwarzian KP equation~\eqref{eq:dSKP} with respect to a special translational symmetry of the equation.
\end{Rem}

\section{Wynn recurrence, circle packings and discrete analytic functions} \label{sec:W-circle}
Geometric interpretation of the \emph{complex} dSKP equation \eqref{eq:dSKP} was first presented in \cite{KoSchief-Men} in the context of the Menelaus theorem and of the Clifford configuration of circles in inversive geometry. One can find there also an equation equivalent to the Wynn recurrence obtained by application of the reduction condition~\eqref{eq:red-SKP-W}. Let us recall their results from our perspective~\cite{DoliwaKosiorek}. In studying the corresponding initial boundary value problem we discuss also geometric meaning of the compatibility of the relevant construction.

In this Section we study the Wynn recurrence in the complex projective line. Such a line, called in this context also the Riemann sphere or the conformal plane, has an additional structure which comes from the standard embedding of the field of real numbers in the complex numbers. The images of the real line under complex-homographic maps are circles or straight lines. By identifying the complex projective line with the (complex) plane supplemented by the infinity point, we identify the special circles passing through that point with straight lines. Homographic transformations preserve the structure (including the angles between circles) and may exchange the infinity point with ordinary ones. In particular, parallel lines are tangent at infinity. 

The conformal plane construction of the point $R_{m,n+1,p+1}$ with the points $R_{m+1,n,p}$, $R_{m,n+1,p}$, $R_{m,n,p+1}$, $R_{m+1,n+1,p}$ and $R_{m+1,n,p+1}$ given reads as follows (see Figure~\ref{fig:Chain-Veblen}):
\begin{itemize}
	\item by $I$ denote the intersection point of the circle passing through the points $R_{m+1,n,p}$, $R_{m,n+1,p}$, and $R_{m+1,n+1,p}$ with the circle passing through the points $R_{m+1,n,p}$, $R_{m,n,p+1}$, and $R_{m+1,n,p+1}$;
	\item the point $R_{m,n+1,p+1}$ is the intersection of the circle passing through the points $I$, $R_{m,n+1,p}$, and $R_{m,n,p+1}$ with the circle passing through the points $I$, $R_{m+1,n+1,p}$, and $R_{m+1,n,p+1}$.
\end{itemize}

\begin{figure}[h!]
	\begin{center}
	\includegraphics[width=12cm]{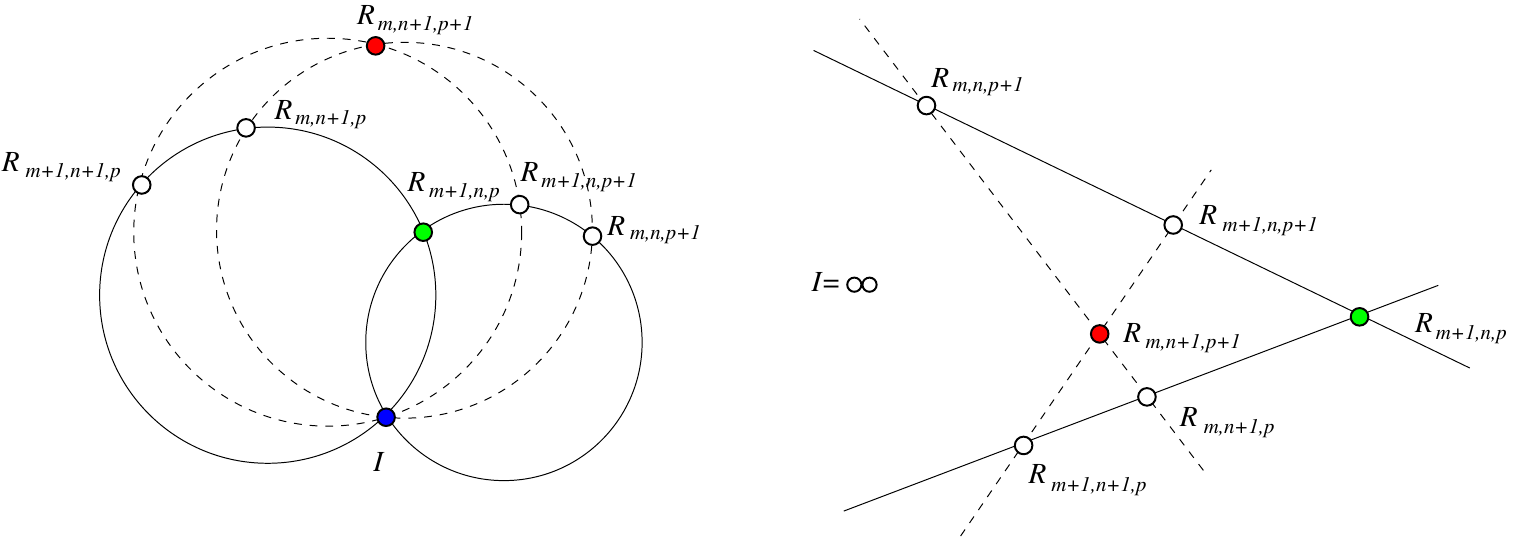}
	\end{center}
	\caption{Circle geometry description of the complex discrete Schwarzian KP equation} 
	\label{fig:Chain-Veblen}
\end{figure} 
\begin{Rem}
	The proof \cite{KoSchief-Men} makes use of the invariance of the discrete Schwarzian KP equation  with respect to the homographies and then, after shifting the intersection point $I$ to the infinity, follows by application of the celebrated Menelaus theorem of the affine geometry~\cite{Coxeter-GR}. 
\end{Rem}

Application of the reduction condition~\eqref{eq:red-SKP-W} forces also the identification $I=R_{m+1,n+1,p}=R_{m,n,p+1}$, where we assume that no other coincidence among the initial points appears. This in turn implies double contact (tangency) of the two pairs of circles in the distinguished point:
\begin{itemize}
	\item the circle passing through the points $R_{m+1,n,p}$, $R_{m,n+1,p}$, and $R_{m+1,n+1,p}$ is tangent at $I$ to the circle passing through the points $I=R_{m+1,n+1,p}$, and $R_{m+1,n,p+1}$,
	\item the circle passing through the points $R_{m+1,n,p}$, $R_{m,n,p+1}$, and $R_{m+1,n,p+1}$ is tangent at $I$ to the circle passing through the points $I=R_{m,n,p+1}$ and $R_{m,n+1,p}$.
\end{itemize} 
The reduced configuration of circles after transition to variables $m,n$ and overall shift (see the proof of Proposition~\ref{prop:Wynn-red}) is visualized in Figure~\ref{fig:Wynn-circles-P}. When the distinguished point $I=R_{m,n}$ is moved to infinity then the pairs of tangent circles become two pairs of parallel lines, and the points $R_{m\pm1,n}, R_{m,n\pm1}$ become vertices of the corresponding parallelogram. This point of view on the two-parameter families of pairwise tangent circles (the so called P-nets) in relation to discrete integrable equations was considered in~\cite{BobenkoPinkall-DSIS}. When the circles intersect orthogonally (all the parallelograms are rectangles) and half of them is removed after the construction (see Figure~\ref{fig:Wynn-circle-init}) then the remaining systems of tangent circles is that considered by Schramm~\cite{Schramm} in the context of discrete complex analysis.

\begin{Rem}
The contemporary interest in circle packings was initiated by Thurston’s rediscovery of the Koebe-Andreev theorem \cite{Koebe} about circle packing realizations of cell complexes of a prescribed combinatorics and by his idea about
approximating the Riemann mapping by circle packings, see \cite{Thurston,Stephenson}. These results
demonstrate surprisingly close analogy to the classical theory and allow one to talk about an
emerging of the ”discrete analytic function theory”, containing the classical theory of analytic
functions as a small circles limit.
Circle patterns with the combinatorics of
the square grid introduced by Schramm~\cite{Schramm}
result in an analytic description, which is closer to the Cauchy-Riemann equations of complex analysis. 
In \cite{BobenkoPinkall-DSIS} the description of Schramm’s square grid circle patterns in conformal setting to an integrable system of Toda type is given. In the same paper it was found that such a system describes a generalization of the Schramm circle patterns, called the P-nets, i.e. discrete conformal maps with the parallelogram property. 
It should be mentioned that the description of the Schramm circle packings in terms of the complex Wynn recurrence can be found in \cite{BobenkoHoffmann}, see equation (15) of the paper, although without any association to the Pad\'{e} theory.
\end{Rem}

\begin{figure}[h!]
	\begin{center}
		\includegraphics[width=12cm]{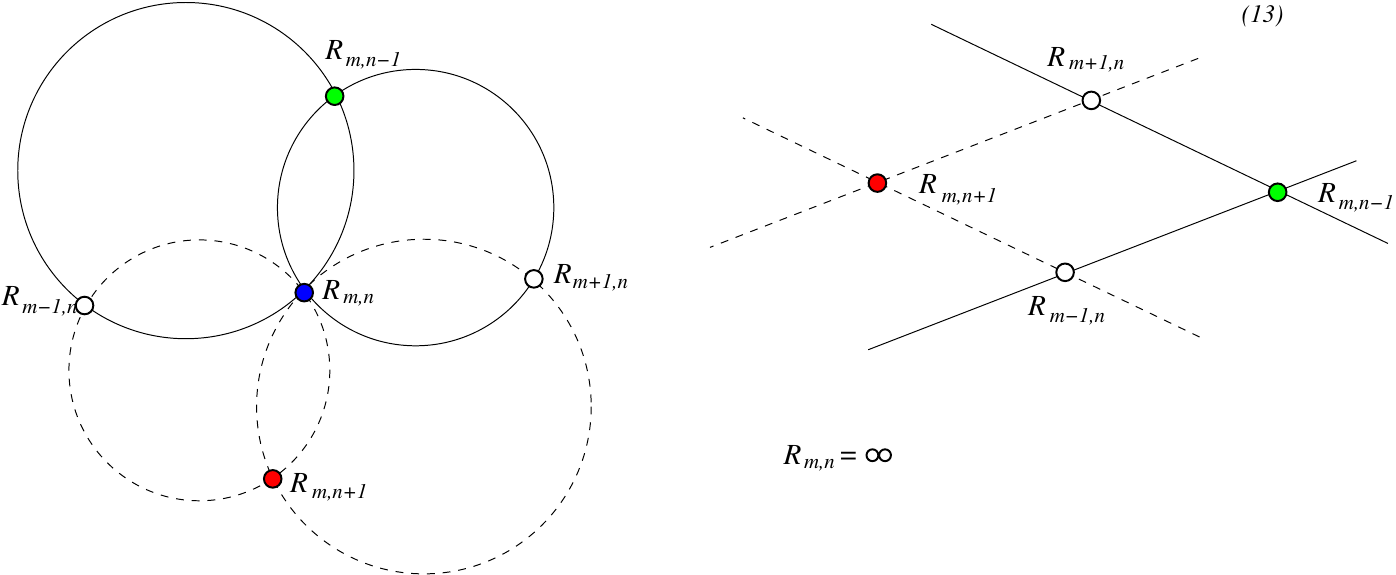}
	\end{center}
	\caption{Two tangent pairs of circles as P-configuration} 
	\label{fig:Wynn-circles-P}
\end{figure} 
The possibility of using the additional geometric (conformal) structure results in the reduction of initial boundary data with respect to the generic Wynn recurrence, as discussed in Section~\ref{sec:Wynn-derivation}. In such reduction it is enough to know $R_{m,0}$, $m\in\mathbb{Z}$, and $R_{0,1}$. Apart from the two "initial" circles: that passing through the points $R_{0,0}$, $R_{1,0}$ and $R_{0,1}$, and  that passing through the points $R_{0,0}$, $R_{-1,0}$ and $R_{0,1}$, the other circles are constructed from two points, with prescribed tangency in one of them. Notice that the description of the lattice states that also at the second point the tangency is required. Such a compatibility of the construction is ensured by the following geometric result. 

\begin{Th}[Tangential Miquel theorem]
	Given four points $P_1, P_2, P_3$ and $P_4$ on the circle $\mathcal{C}$. Let 
	\begin{itemize}
		\item $\mathcal{C}_1$ be a circle passing through $P_1$ and $P_2$,
		\item $\mathcal{C}_2$ the circle tangent to $\mathcal{C}_1$ at $P_2$ and passing through $P_3$,
		\item $\mathcal{C}_3$ the circle tangent to $\mathcal{C}_2$ at $P_3$ and passing through $P_4$, 
		\item $\mathcal{C}_4$ the circle tangent to $\mathcal{C}_3$ at $P_4$ and passing through $P_1$.
	\end{itemize}  
Then the circle $\mathcal{C}_4$ is also tangent to $\mathcal{C}_1$ at $P_1$. 
\end{Th}
\begin{figure}
	\begin{center}
		\includegraphics[width=12cm]{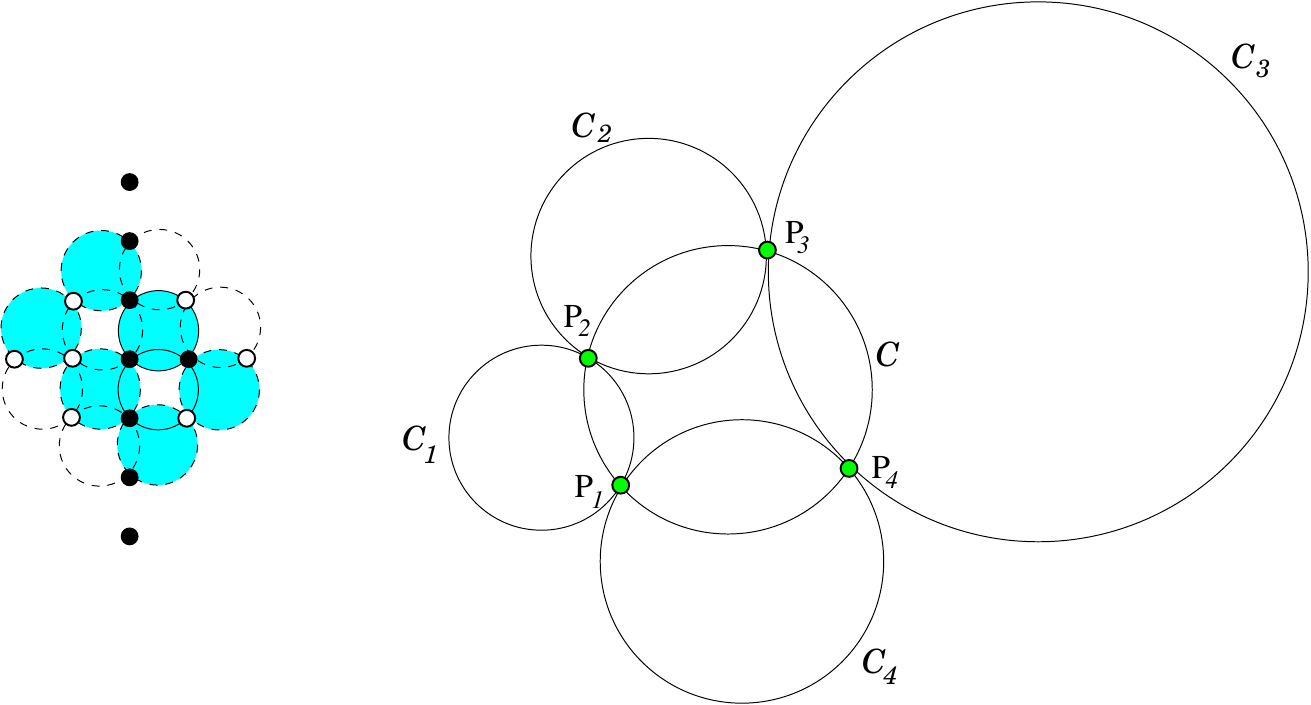}
	\end{center}
	\caption{Construction of the circle lattice from the initial data represented by the black points (on the left) and the tangential Miquel theorem implying consistency of the construction (on the right)} 
	\label{fig:Wynn-circle-init}
\end{figure} 
The proof is elementary, but it helps to shift one of the points (for example $P_1$) to infinity.
\begin{Rem}
	The above result is a limiting version of the six-circles Miquel theorem~\cite{Herzer} which was used in~\cite{CDS} to prove consistency of the circular reduction with the geometric integrable construction of the multidimensional lattice of planar quadrilaterals~\cite{MQL}.
\end{Rem}

\bibliographystyle{amsplain}

\section{Conclusion and open problems}
Motivated by the application of the discrete-time Toda chain equations in the theory of Pad\'{e} approximants, as one of the Frobenius identities, we studied the integrability of the Wynn recurrence. We investigated non-commutative version of the Pad\'{e} theory using quasideterminants. The proper form of the corresponding non-commutative discrete-time Toda chain equations and of their linear problem was obtained in close analogy with the standard~\cite{Baker,Gragg} determinantal approach to Frobenius identities, whose non-commutative versions were presented as well. We provided, on example of the Fibonacci language, the problem of rational approximations in the theory  of formal languages, where the non-commutativity of symbols that build the series is inherent and cannot be discarded.

After deriving the non-commutative Wynn recurrence from the linear problem of the discrete-time Toda chain equations, we gave also its second derivation as a reduction of the discrete non-commutative Schwarzian Kadomtsev--Petviashvili equation, which allowed to discover the geometric construction behind the Wynn recurrence valid for arbitrary skew field/division ring. It turns out that the same reduction in the complex field case~\cite{KoSchief-Men} gives the circle packings relevant in the theory of "discrete complex analysis"~\cite{Schramm}. It is therefore remarkable that two different approximation schemata of complex analytic functions: by rational functions and by "discrete analytic functions", are described by the same integrable equation. 

As we mentioned in the Introduction, in the literature there are known also other problems of numerical analysis related to integrable systems. In addition to looking for new examples of such a relationship, one can ask about the general reason explaining its existence.  
The editor of series Mathematics and Its Applications writes in~\cite{NNMRA}:
\emph{Rational or Pad\'{e} approximation [...] is still something of a mystery to this editor. Not the basic idea itself, which is lucid enough. But why is the technique so enormously efficient, and numerically useful, in so many fields ranging from physics to electrical engineering with continued fractions, orthogonal polynomials, and completely integrable systems tossed in for good measure.}

 Having in mind the special role played in the theory of integrable systems by Hirota's discrete KP equation~\cite{Hirota-2dT,KNS-rev,Zabrodin} one can start such an investigation by finding related problem which should contain existing examples as special cases. Because of the symmetry structure of the equation it is desirable that such problem would allow for arbitrary number of dimensions of discrete parameters. A good candidate is provided by the so called Hermite--Pad\'{e} approximation problem, and the research in this direction will be reported in a separate publication~\cite{Doliwa-Siemaszko-HP}. 

\subsection*{Acknowledgements} The authors would like to thank the reviewers for their constructive comments, which allowed for the improvement of the presentation of the results.

\subsection*{Data availability}
Data sharing not applicable to this article as no datasets were generated or analyzed during the current study.

\subsection*{Conflict of interest}
The authors declare that they have no conflict of interest.

\providecommand{\bysame}{\leavevmode\hbox to3em{\hrulefill}\thinspace}

\end{document}